\documentclass{article} 
\usepackage[colorlinks=true]{hyperref}
\usepackage{amsmath}
\usepackage[lmargin=3cm, rmargin=3cm]{geometry}
\usepackage{amsthm}
\usepackage{bbm}
\usepackage{amsfonts}
\usepackage{amssymb}
\usepackage{graphicx}
\usepackage{lmodern}
\usepackage{url}
\usepackage{subfigure}

\usepackage{notations}

\newcommand{\bfx}{\mathbf{x}}
\newcommand{\bfu}{\mathbf{u}}

\newcommand{\bfz}{\mathbf{z}}
\newcommand{\bff}{\mathbf{f}}
\newcommand{\bfL}{\mathbf{L}}
\newcommand{\bfQ}{\mathbf{Q}}

\def\cB{\mathcal{B}}
\def\bQ{\mathbf{Q}}
\def\bPsi{\mathbf{\Psi}}

\def\tr{\textnormal{tr}}
\renewcommand{\IND}{\ensuremath{\mathbbm{1}}}
\def\twofig{0.4\textwidth}

\def\figdir{.}

\theoremstyle{plain}
\newtheorem{theo}{Theorem}[]
\newtheorem{prop}[theo]{Proposition}

\newtheorem{theorem}{Theorem}[section]
\newtheorem{lemma}[theorem]{Lemma}

\newcommand{\bftheta}{\boldsymbol{\theta}}

\def\cL{\mathcal{L}}

\def\cQ{\mathcal{Q}}
\def\cI{\mathcal{I}}

\def\bpo{\big(\big(}
\def\bpc{\big)\big)}

\newif\iflong
\longtrue

\usepackage[]{trackchanges}
\renewcommand{\initialsOne}{MT}
\renewcommand{\initialsTwo}{RB}

\title{Inference for determinantal point processes \\without spectral knowledge}

\author{
R\'emi Bardenet$^*$ \\
CNRS \& CRIStAL, Univ. Lille\\
59655 Villeneuve d'Ascq, France\\
\texttt{remi.bardenet@gmail.com} \\
\and
Michalis K. Titsias$^*$\\
Dept. of Informatics, Athens Univ. of Economics and Business\\
Patision 76, 10434 Athens, Greece \\
\texttt{mtitsias@aueb.gr}\\
~\\
\footnotesize{$^*$Both authors contributed equally to this work.}
}

\date{}

\begin{document}
\maketitle 

\begin{abstract}
Determinantal point processes (DPPs) are point process models that
naturally encode diversity between the points of a
given realization, through a positive definite kernel $K$. DPPs possess desirable properties, such as exact
sampling or analyticity of the moments, but learning the parameters of
kernel $K$ through likelihood-based inference is not
straightforward. First, the kernel that appears in the
likelihood is not $K$, but another kernel $L$ related to $K$ through
an often intractable spectral decomposition. This issue is typically bypassed in machine learning by
directly parametrizing the kernel $L$, at the price of some
interpretability of the model parameters. We follow this approach
here. Second, the likelihood has an intractable normalizing
constant, which takes the form of a large determinant in the case of a
DPP over a finite set of objects, and the form of a Fredholm determinant in the
case of a DPP over a continuous domain. Our main contribution is to derive bounds on the likelihood of
a DPP, both for finite and continuous domains. Unlike previous work, our bounds are
cheap to evaluate since they do not rely on approximating the spectrum
of a large matrix or an operator. Through usual arguments, these bounds thus yield cheap variational
inference and moderately expensive exact Markov chain Monte Carlo inference methods for DPPs.
\end{abstract}

\section{Introduction}
\label{s:intro}
Determinantal point processes (DPPs) are point processes \cite{DaVe03} that
encode repulsiveness using algebraic arguments. They first appeared in
\cite{Mac75}, and have since then received much attention, as they
arise in many fields, e.g. random matrix theory, combinatorics,
quantum physics. We refer the reader to \cite{KuTa12, LaMoRuSub, HKPV06} for detailed
tutorial reviews, respectively aimed at audiences of machine learners, statisticians,
and probabilists. More recently, DPPs have been considered as a
modelling tool, see e.g. \cite{LaMoRuSub, KuTa12, ZoAd12}: DPPs
appear to be a natural alternative to Poisson processes
when realizations should exhibit repulsiveness. In \cite{KuTa12}, for
example, DPPs are used to model diversity among summary timelines in a
large news corpus. In \cite{AFAT14}, DPPs model diversity among the
results of a search engine for a given query. In \cite{LaMoRuSub}, DPPs model the spatial
repartition of trees in a forest, as similar trees compete for nutrients in the
ground, and thus tend to grow away from each other. With these
modelling applications comes the question of learning a DPP from data,
either through a parametrized form \cite{LaMoRuSub,AFAT14}, or
non-parametrically \cite{GKFT14, MaSr15}. We focus in this paper on
parametric inference.

Similarly to the correlation between the function values in a Gaussian
process \cite{RaWi06}, the repulsiveness in a DPP is defined through a
kernel $K$, which measures how much two points in a realization repel
each other. The likelihood of a DPP involves the evaluation and the spectral decomposition
of an operator $\cL$ defined through a kernel $L$ that is related to $K$. There are two main issues that
arise when performing likelihood-based inference for a DPP. First, the
likelihood involves evaluating the kernel $L$, while it is more
natural to parametrize $K$ instead, and there is no easy link between
the parameters of these two kernels. The second issue is
that the spectral decomposition of the operator $\cL$ required in the likelihood evaluation is rarely available in
practice, for computational or analytical reasons. For example, in the case of a
large finite set of objects, as in the news corpus application
\cite{KuTa12}, evaluating the likelihood once requires the
eigendecomposition of a large matrix. Similarly, in the case of a
continuous domain, as for the forest application \cite{LaMoRuSub}, the
spectral decomposition of the operator $\cL$ may
not be analytically tractable for nontrivial choices of kernel $L$. In this paper, we focus on the second
issue, i.e., we provide likelihood-based inference methods that assume
the kernel $L$ is parametrized, but that do not
require any eigendecomposition, unlike \cite{AFAT14}. More specifically, our main
contribution is to provide bounds on the likelihood of a DPP that do
not depend on the spectral decomposition of the operator $\cL$. For
the finite case, we draw inspiration from bounds used for variational
inference of Gaussian processes \cite{titsias}, and we extend these
bounds to DPPs over continuous domains.

For ease of presentation, we first consider DPPs over finite sets of
objects in Section~\ref{s:finite}, and we derive bounds on the
likelihood. In Section~\ref{s:learning}, we plug these bounds into
known inference paradigms: variational inference and Markov chain
Monte Carlo inference. In Section~\ref{s:continuous}, we extend our
results to the case of a DPP over a continuous domain. Readers who are
only interested in the finite case, or who are unfamiliar with
operator theory, can safely skip Section~\ref{s:continuous} without
missing our main points. In
Section~\ref{s:experiments}, we experimentally validate our results,
before discussing their breadth in Section~\ref{s:discussion}.

\section{DPPs over finite sets}
\label{s:finite}

\subsection{Definition and likelihood}
\label{s:Finite:def}

Consider a discrete set of items $\mathcal{Y} = \{\bfx_1, \ldots,
\bfx_n\}$, where $\bfx_i \subset \mathbb{R}^d$ is a vector of attributes that describes item
$i$. Let $K$ be a symmetric positive definite kernel \cite{CrSh04} on $\mathbb{R}^d$, and
let $\bK = ((K(\bx_i,\bx_j)))$ be the Gram matrix of $K$. The DPP of
kernel $K$ is defined as the probability distribution over all possible 
$2^n$ subsets $Y \subseteq \mathcal{Y}$ such that
\begin{equation}
\mathbb{P}(A\subset Y) = \det(\bK_A),
\label{e:discreteCDF}
\end{equation}
where $\bK_A$ denotes the sub-matrix of $\bK$ indexed by the elements of
$A$. This distribution exists and is unique if and only if the
eigenvalues of $\bK$ are in $[0,1]$ \cite{HKPV06}. Intuitively, we can think of
$K(\bx,\by)$ as encoding the amount of negative correlation, or
``repulsiveness'' between $\bx$ and $\by$. Indeed, as remarked in
\cite{KuTa12}, \eqref{e:discreteCDF} first yields that diagonal elements of
$\bK$ are marginal probabilities: $\mathbb{P}(\bx_i\in Y) =
K_{ii}$. Equation \eqref{e:discreteCDF} then entails that $\bx_i$ and $\bx_j$ are
likely to co-occur in a realization of $Y$ if and only if
$$
\det K_{\{\bx_i,\bx_j\}} = K(\bx_i,\bx_i)K(\by_i,\by_i) -
K(\bx_i,\bx_j)^2 = \mathbb{P}(\bx_i\in Y) \mathbb{P}(\bx_j\in Y) -  K_{ij}^2
$$
is large: off-diagonal terms in $\bK$ indicate whether points tend to
co-occur.

Providing the eigenvalues of $\bK$ are further restricted to be in
$[0,1)$, the DPP of kernel $K$ has a likelihood \cite{DaVe03}. More
specifically, writing $Y_1$ for a realization of $Y$,
\begin{equation}
\mathbb{P}(Y=Y_1) = \frac{\det \bL_{Y_1}}
{\text{det}(\bL + \bI)},
\label{e:discreteDensity}
\end{equation}
where $\bL=(\bI-\bK)^{-1}\bK$, $\bI$ is the $n\times n$ identity
matrix, and $\bL_{Y_1}$ denotes the sub-matrix of $\bL$ indexed by the
elements of $Y_1$. Now, given a realization $Y_1$, we would like to
infer the parameters of kernel $K$, say the parameters $\theta_K =
(a_K,\sigma_K)\in (0,\infty)^2$ of a squared exponential kernel \cite{RaWi06}
$$
K(\bx,\by) = a_K \exp\left( - \frac{\Vert \bx -\by\Vert^2}{2\sigma_K^2}\right).
$$
Since the trace of $\bK$ is the expected number of points in $Y$
\cite{HKPV06}, one can estimate $a_K$ by the number of points in the data
divided by $n$ \cite{LaMoRuSub}. But $\sigma_K$, the parameter governing
the repulsiveness, has to be fitted. If the number of items
$n$ is large, likelihood-based methods such as maximum likelihood are too
costly: each evaluation of \eqref{e:discreteDensity} requires $\cO(n^2)$ storage and $\cO(n^3)$
time. Furthermore, valid choices of $\theta_K$ are constrained, since
one needs to make sure the eigenvalues of $\bK$ remain in $[0,1)$.

A partial work-around is to note that given any symmetric positive definite kernel $L$, the likelihood
\eqref{e:discreteDensity} with matrix $\bL = ((L(\bx_i,\bx_j)))$ corresponds to a valid choice of $K$, since
the corresponding matrix $\bK = \bL(\bI+\bL)^{-1}$ necessarily has
eigenvalues in $[0,1]$, which makes sure the DPP exists
\cite{HKPV06}. The work-around consists in directly parametrizing and inferring the kernel $L$ instead
of $K$, so that the numerator of \eqref{e:discreteDensity} is cheap to
evaluate, and parameters are less constrained. Note that this step
favours tractability over interpretability of the
inferred parameters: if we assume 
$$
L(\bx,\by) = a_L \exp\left( - \frac{\Vert \bx -\by\Vert^2}{2\sigma_L^2}\right),
$$
the number of points and the repulsiveness of the points
in $Y$ do not decouple as nicely as when $K$ is squared
exponential. For example, the expected number of
items in $Y$ depends on $a_L$ and $\sigma_L$ now, and both parameters
also affect repulsiveness. There is some work investigating
approximations to $\bK$ to retain the more interpretable
parametrization \cite{LaMoRuSub}, but the machine learning literature \cite{KuTa12, AFAT14} almost exclusively
adopts the more tractable parametrization of $L$. In this paper, we also make
this choice of parametrizing $L$ directly.

Now, the only expensive step in the evaluation of
\eqref{e:discreteDensity} is the computation of
$\det(\bL+\bI)$. While this still prevents the application of maximum
likelihood, bounds on this determinant can be used in a variational
approach or an MCMC algorithm, for example. In \cite{AFAT14}, bounds
on $\det(\bL+\bI)$ are proposed, requiring only the first $m$ eigenvalues of $\bL$, where
$m$ is chosen adaptively at each MCMC iteration to make the acceptance decision
possible. This still requires the application of power iteration methods,
which are limited to the finite domain case, require storing the
whole $n\times n$ matrix $\bL$, and are prohibitively slow when the number of required
eigenvalues $m$ is large.

\subsection{Nonspectral bounds on the likelihood}

Let us denote by $\bL_{\cA\cB}$ the submatrix of $\bL$ where row indices
correspond to the elements of $\cA$, and column indices to those of
$\cB$. When $\cA=\cB$, we simply write $\bL_\cA$ for
$\bL_{\cA\cA}$, and we drop the subscript when $\cA=\cY$. Drawing inspiration from sparse approximations to
Gaussian processes using inducing variables
\cite{titsias}, we let $\cZ=\{\bz_1,\dots,\bz_m\}$ be an
arbitrary set of points in $\mathbb{R}^d$, and we approximate $\bL$ by
$ \bfQ = \bL_{\cY\cZ}[\bL_\cZ]^{-1}\bL_{\cZ\cY}$. Note that we do not constrain $\cZ$ to belong to $\cY$, so that our bounds do not rely on a
Nystr\"om-type approximation \cite{AKFT13}. We term $\cZ$ ``pseudo-inputs'', or ``inducing inputs''.

\begin{prop}
\begin{equation}
\frac{1}
{\det(\bQ + \bI)} 
e^{-  
\tr\left(\bL  -\bQ \right)}
 \leq \frac{1}
{\det(\bL + \bI)}  \leq 
\frac{1}
{\det(\bQ + \bI)}. 
\label{eq:boundDetFinite}
\end{equation}
\label{p:boundsFinite}
\end{prop}

\iflong
\begin{proof}
The right inequality is a straightforward consequence of the Schur
complement $\bL-\bQ$ being positive
semidefinite. For instance, once could remark that for $\bv\in\mathbb{R}^n$,
$$
\bv^T\bL\bv \leq 1 \Rightarrow \bv^T\bQ\bv\leq 1,
$$
so that
$$
\int_{\mathbb{R}^n} \IND_{\{\bv^T(\bL+\bI)\bv\leq 1\}}d\bv \leq \int_{\mathbb{R}^n} \IND_{\{\bv^T(\bQ+\bI)\bv\leq 1\}}d\bv.
$$
A change of variables using the Cholesky decompositions of $\bL+\bI$
and $\bQ+\bI$ yields the desired inequality. 

The left inequality in \eqref{eq:boundDetFinite} can be proved along the lines of \cite{titsias}, using
variational arguments (as will be discussed in detail in Section \ref{sec:svi}). Next we give an alternative, more direct
proof based on an inequality on determinants \cite[Theorem
1]{SeSi75}. For any real symmetric matrix $A = P \diag(\lambda_i) P^T$, define its
absolute value as $\vert A\vert = P
\text{diag}(\vert\lambda_i\vert)P^T$. In particular, for a
positive semidefinite $A$, $\vert A\vert = A$. Applying \cite[Theorem 1]{SeSi75} and noting that $\bL$, $\bQ$ and $\bL-\bQ$ are
positive semidefinite, it comes
\begin{eqnarray}
\det(\bL+\bI) &=& \text{det}(\bfL  - \bfQ  + \bfQ  + \bI) \nonumber\\
&\leq& \text{det}(\vert\bfL - \bfQ\vert +  \bI)
\text{det}(\vert\bfQ\vert + \bI)\nonumber\\
&=& \text{det}(\bfL - \bfQ +  \bI)
\text{det}(\bfQ + \bI)
\label{e:tool1}
\end{eqnarray}
Now, denote by $\tilde\lambda_i$ the eigenvalues of $\bL-\bQ$, which
are all nonnegative. It comes
\begin{equation}
\text{det}(\bfL - \bfQ +  \bI) = \prod_{i=1}^n(1+\tilde\lambda_i) \leq
\prod_{i=1}^n e^{\tilde\lambda_i} =  e^{\text{tr}\left(\bfL - \bfQ\right)},
\label{e:tool2}
\end{equation}
where we used the inequality $1+x\leq e^x$. Plugging \eqref{e:tool2}
into \eqref{e:tool1} yields the left part of \eqref{eq:boundDetFinite}. 
\end{proof}
\fi

\section{Learning a DPP using bounds}
\label{s:learning}

In this section, we explain how to run variational inference and
Markov chain Monte Carlo methods using the bounds in Proposition~\ref{p:boundsFinite}. In this section, we also make
connections with variational sparse Gaussian processes more explicit.

\subsection{Variational inference}
\label{s:variational}
The lower bound in Proposition~\ref{p:boundsFinite} can be used for
variational inference. Assume we have $T$ point process realizations
$Y_1,\dots,Y_T$, and we fit a DPP with kernel $L=L_\theta$. The log likelihood can be expressed using \eqref{e:discreteDensity}
\begin{equation}
\ell(\theta) = \sum_{i=1}^T \log \text{det}(\bL_{Y_t}) - T\log\det(\bL+\bI).
\label{e:likelihood}
\end{equation}
Let $\cZ$ be an arbitrary set of $m$ points in $\mathbb{R}^d$. Proposition~\ref{p:boundsFinite} then yields a lower bound
\begin{equation}
\mathcal{F}(\theta, \cZ) \defeq \sum_{t=1}^T \log \text{det}(\bL_{Y_t}) - T \log \text{det}(\bQ + \bI) 
+ T \text{tr}\left(\bL - \bQ \right) \leq \ell(\theta).
\label{eq:boundLogLikelihood}
\end{equation}
The lower bound $\mathcal{F}(\bftheta, \cZ)$ can be computed
efficiently in $O(n m^2)$ time. Instead of maximizing
$\ell(\theta)$, we can maximize $\mathcal{F}(\theta, \cZ)$ jointly with respect to the kernel parameters $\theta$ and the
variational parameters $\cZ$. 

To maximize \eqref{eq:boundLogLikelihood}, one can e.g. implement an EM-like scheme,
alternately optimizing in $\cZ$ and $\btheta$. Kernels are often differentiable with respect to $\theta$, and sometimes $\mathcal{F}$
will also be differentiable with respect to the pseudo-inputs $\cZ$, so that
gradient-based methods can help. In the general case, black-box
optimizers such as CMA-ES \cite{Han06}, can also be employed.

\iflong
\subsubsection{Connections with variational inference in sparse
  GPs \label{sec:svi}}
\label{s:svi}
We now discuss the connection between the variational lower bound in
Proposition~\ref{p:boundsFinite} and the variational lower bound used
in sparse Gaussian process (GP; \cite{RaWi06}) models
\cite{titsias}. Although we do not explore this in the current paper,
this connection could extend the repertoire of variational inference
algorithms for DPPs by including, for instance, 
stochastic optimization variants.  

Assume function $f$ follows a GP distribution with zero mean function and kernel
function $L$, so that the vector $\bff$ of function
values evaluated at $\mathcal{Y}$ follows the Gaussian distribution     
$\mathcal{N}({\bf f}| {\bf 0}, \bL)$. Then, the standard Gaussian integral yields\footnote{Notice that 
$\int \mathcal{N}({\bf f}| {\bf 0}, \bL) e^{ - \frac{1}{2} {\bf f}^T {\bf f} }   d \bff 
= (2 \pi)^{n/2} \int \mathcal{N}({\bf f}| {\bf 0}, \bL)
\mathcal{N}({\bf 0} | \bff, {\bf I}) d \bff = (2 \pi)^{n/2}
\mathcal{N}({\bf 0}|{\bf 0}, \bL + {\bf I})$, and \eqref{e:marginal} follows.} 
\begin{equation}
 \frac{1}
 {\text{det}(\bL + \bI)}
 = \left( \int \mathcal{N}({\bf f}| {\bf 0}, \bL) e^{ - \frac{1}{2} {\bf f}^T {\bf f} }   d \bff \right)^2. 
\label{e:marginal}
\end{equation}
Following this interpretation, we augment the vector $\bff$ with a vector of $m$ extra auxiliary function values $\bfu$, referred to as    
inducing variables, evaluated at the inducing points $\cZ$ so that jointly $(\bff, \bf u)$ 
follows
\begin{align}
p(\bff,\bfu) & = \cN\left(0, \begin{pmatrix}  \bL & \bL_{\cY\cZ} \\
    \bL_{\cZ\cY} & \bL_\cZ\end{pmatrix}\right), \nonumber \\
    & =  \mathcal{N}(\bff| \bL_{\cY\cZ}[\bL_\cZ]^{-1}\bfu,
\bL - \bfQ) \mathcal{N}(\bfu|{\bf 0}, \bL_\cZ).
\end{align}
Now by using the fact that $\mathcal{N}({\bf f}| {\bf 0}, \bL) = \int p(\bff,\bfu) d \bfu$, the integral in \eqref{e:marginal} can be expanded so that  
\begin{equation}
 \frac{1}
 {\text{det}(\bL +\bI)}
 = \left( \int  \mathcal{N}(\bff| \bL_{\cY\cZ}[\bL_\cZ]^{-1}\bfu, \bL - \bfQ) \mathcal{N}(\bfu|{\bf 0}, \bL_\cZ)  e^{ - \frac{1}{2} \bff^T \bff }  d \bff d \bfu \right)^2. 
 \label{eq:gaussReprDiscrete}
\end{equation}
We can bound the above integral using
Jensen's inequality and the variational distribution $q(\bff,\bfu) = \mathcal{N}(\bff|
\bL_{\cY\cZ}[\bL_\cZ]^{-1}\bfu, \bL - \bfQ) q(\bfu)$, where $q(\bfu)$ is a marginal variational distribution over the 
inducing variables $\bfu$. This form of variational distribution is
exactly the one used for sparse GPs \cite{titsias}, and by treating the factor $q(\bfu)$ optimally we can recover the left
lower bound in Proposition~\ref{p:boundsFinite}, following the lines
of \cite{titsias}. We provide details in Appendix A.
 
The above connection suggests that much of the technology developed for speeding up  GPs can be transferred to DPPs. 
For instance, if we explicitly represent the $q(\bfu)$ variational distribution in the above formulation, then we can develop 
stochastic variational inference variants for learning DPPs based on data
subsampling \cite{Hoffman:2013}. In other words, we can apply to DPPs stochastic variational inference
algorithms for sparse GPs such as \cite{HensmanFL13}. 
\fi

\subsection{Markov chain Monte Carlo inference}
\label{s:MH}
If approximate inference is not suitable, we can use the bounds in
Proposition~\ref{p:boundsFinite} to build a more expensive Markov chain Monte Carlo \cite{RoCa04}
sampler. Given a prior distribution $p(\theta)$ on the parameters $\theta$ of
$L$, Bayesian inference relies on the posterior distribution $\pi(\theta) \propto \exp(\ell(\theta)) p(\theta)$,
where the log likelihood $\ell(\theta)$ is defined in
\eqref{e:likelihood}. A standard approach to sample approximately from $\pi(\theta)$ is
the Metropolis-Hastings algorithm (MH; \cite[Chapter 7.3]{RoCa04}).
MH consists in building an ergodic Markov chain of invariant distribution
$\pi(\theta)$. Given a proposal $q(\theta'\vert\theta)$, the MH
algorithm starts its chain at a user-defined $\theta_{0}$, then at
iteration $k+1$ it proposes a candidate state $\theta'\sim q(\cdot\vert\theta_{k})$
and sets $\theta_{k+1}$ to $\theta'$ with probability 
\begin{eqnarray}
\alpha(\theta_{k},\theta') & = & \min\left[1, \frac{e^{\ell(\theta')} p(\theta')}{e^{\ell(\theta_k)} p(\theta_k)}\frac{q(\theta_{k}\vert\theta')}{q(\theta'\vert\theta_{k})}\right]
\label{e:MHAcceptanceRatio}
\end{eqnarray}
while $\theta_{k+1}$ is otherwise set to $\theta_{k}$. The core of the
algorithm is thus to draw a Bernoulli variable with parameter
$\alpha=\alpha(\theta,\theta')$ for $\theta,\theta'\in\mathbb{R}^d$. This is typically implemented by drawing a
uniform $u\sim\cU_{[0,1]}$ and checking whether $u<\alpha$. In our DPP application, we cannot evaluate
$\alpha$. But we can use Proposition~\ref{p:boundsFinite} to
build a lower and an upper bound $ \ell(\theta)\in [b_-(\theta,\cZ),b_+(\theta,\cZ)]$,
which can be arbitrarily refined by increasing the cardinality of $\cZ$
and optimizing over $\cZ$. We can thus build a lower and upper bound
for $\alpha$
\begin{equation}
\label{e:boundsOnAlpha}
b_-(\theta',\cZ') - b_+(\theta,\cZ) +
\log\left[\frac{p(\theta')}{p(\theta)}\right] \leq
\log\alpha \leq b_+(\theta',\cZ') - b_-(\theta,\cZ) +
\log\left[\frac{p(\theta')}{p(\theta)}\right].
\end{equation}
Now, another way to draw a Bernoulli variable with parameter $\alpha$
is to first draw $u\sim\cU_{[0,1]}$, and then refine the bounds in
\eqref{e:boundsOnAlpha}, by augmenting the numbers $\vert\cZ\vert$, $\vert\cZ'\vert$ of inducing
variables and optimizing over $\cZ,\cZ'$, until $\log u$ is out of the interval formed by the
bounds in \eqref{e:boundsOnAlpha}. Then one can decide whether
$u<\alpha$. This Bernoulli trick is sometimes named {\it retrospective sampling}
and has been suggested as early as \cite{Dev86}. It has been used
within MH for inference on DPPs with spectral bounds
in \cite{AFAT14}, we simply adapt it to our non-spectral bounds.

\section{The case of continuous DPPs}
\label{s:continuous}

DPPs can be defined over very general spaces \cite{HKPV06}. We
limit ourselves here to point processes on $\cX\subset \mathbb{R}^d$
such that one can extend the notion of likelihood. In particular, we define here a DPP on $\cX$ as in \cite[Example
5.4(c)]{DaVe03}, by defining its Janossy density. For definitions of
traces and determinants of operators, we follow \cite[Section VII]{GoGoKa90}.

\subsection{Definition}

Let $\mu$ be a measure on $(\mathbb{R}^d,\cB(\mathbb{R}^d))$ that is
continuous with respect to the Lebesgue measure, with density
$\mu'$. Let $L$ be a symmetric positive
definite kernel. $L$ defines a self-adjoint operator on $L^2(\mu)$ through
$
\cL(f) \defeq \int L(\bx,\by)f(\by) d\mu(\by).
$
Assume $\cL$ is trace-class, and
\begin{equation}
\text{tr}(\cL) = \int_{\mathcal{X}} L(\bfx,\bfx) d\mu(\bfx).  
\label{eq:traceL}
\end{equation}
We assume \eqref{eq:traceL} to avoid
  technicalities. Proving \eqref{eq:traceL} can be done by requiring
  various assumptions on $L$ and $\mu$. Under the
  assumptions of Mercer's theorem, for instance, \eqref{eq:traceL}
  will be satisfied \cite[Section VII, Theorem 2.3]{GoGoKa90}. More
  generally, the assumptions of \cite[Theorem 2.12]{Sim05} apply to
  kernels over noncompact domains, in particular the Gaussian kernel
  with Gaussian base measure that is often used in practice. We denote by $\lambda_i$ the eigenvalues of the compact operator
$\cL$. There exists \cite[Example 5.4(c)]{DaVe03} a simple\footnote{i.e.,
for which all points in a realization are distinct.} point process on
$\mathbb{R}^d$ such that
\begin{equation}
\mathbb{P}\begin{pmatrix}\text{There are $n$ particles, one in
      each of}\\
\text{the infinitesimal balls $B(\bx_i, d\bx_i)$}\end{pmatrix} =
\frac{\det ((L(\bx_i,\bx_j))}{\det(\cI+\cL)} \mu'(x_1)\dots \mu'(x_n),
\label{eq:continuousDensity}
\end{equation}
where $B(\bx, r)$ is the open ball of center $\bx$ and radius $r$, and
where $\det(\cI+\cL) \defeq \prod_{i=1}^{\infty} (\lambda_i + 1)$
is the Fredholm determinant of operator $\cL$ \cite[Section VII]{GoGoKa90}. Such a process is called the determinantal point process
associated to kernel $L$ and base measure $\mu$.\footnote{There is a
  notion of kernel $K$ for general DPPs \cite{HKPV06}, but we define $L$ directly here, for the sake of simplicity. The
  interpretability issues of using $L$ instead of $K$ are the same as
  for the finite case, see Sections~\ref{s:finite} and \ref{s:experiments}.} Equation~\eqref{eq:continuousDensity} is the continuous
equivalent of \eqref{e:discreteDensity}. Our bounds
require $\bPsi$ to be computable. This is the case for the popular
Gaussian kernel with Gaussian base measure.

\subsection{Nonspectral bounds on the likelihood}
In this section, we derive bounds on the likelihood
\eqref{eq:continuousDensity} that do not require to compute the
Fredholm determinant $\det(\cI+\cL)$. 

\begin{prop}
Let $\cZ = \{\bz_1,\dots,\bz_m\}\subset\mathbb{R}^d$, then
\begin{equation}
\frac{\det \bL_\cZ}
{\det(\bL_\cZ + \bPsi)}  e^{- \int L(\bfx,\bfx) d \mu(\bfx) + \text{tr}(\bL_\cZ^{-1} \bPsi) } 
 \leq \frac{1}
{\det(\cI+\cL)}  \leq \frac{\det \bL_\cZ }
{\det(\bL_\cZ + \bPsi)}, 
\label{eq:boundDetContinuous}
\end{equation}
where $\bL_\cZ = ((L(\bz_i,\bz_j))$ and 
$
\bPsi_{ij} = \int L(\bfz_i, \bfx) L(\bfx,\bfz_j) d\mu(\bfx).
$
\label{p:boundsContinuous}
\end{prop}

\iflong
To see how similar \eqref{eq:boundDetContinuous} is to
\eqref{eq:boundDetFinite}, we define $\cQ_\cZ$ to be the operator
on $L^2(\mu)$ associated to the kernel
\begin{equation}
Q_\cZ(\bfx,\bfx') = L(\bfx,\cZ) \bL_{\cZ}^{-1} L(\cZ,\bfx'),
\label{e:defNystromKernel}
\end{equation}
where 
$$ L(\bfx,\cZ) = L(\cZ, \bfx)^T \defeq \begin{pmatrix} L(\bfx, \bz_1)
  &\dots& L(\bfx, \bz_m)\end{pmatrix}.$$
Then the following extension of the matrix-determinant lemma shows that the common
factor in the left and right hand side of \eqref{eq:boundDetContinuous} is the inverse of $\det(\cI+\cQ_\cZ)$,
as in \eqref{eq:boundDetFinite}.

\begin{lemma}
\label{l:matrixDeterminantLemmaForOperators} With the notation of Section~\ref{s:continuous}, it holds
$$ \det (\cI + \cQ_\cZ) = \frac{\det ( \bL_\cZ + \bPsi)}{\det \bL_\cZ}.$$
\end{lemma}
\begin{proof}
First note that $\cQ_\cZ$ has finite rank since for $f\in L^2(\mu)$,
$$ \cQ_\cZ f = \sum_{u,v=1}^M [\bL_{\cZ}^{-1}]_{uv}
L(\bz_u,\cdot) \int L(\bz_v,y)f(y) \, d\mu(y) \in S$$
with 
$$S = \text{Span}\left( L(\bz_i,\cdot);
  1\leq i\leq M\right).$$
Note also that the $L(\bz_i, \cdot)$'s are linearly independent since $L$
is a positive definite kernel. Now let $(\phi_i)_{1\leq i\leq M}$ be an
orthonormal basis of $S$, i.e. $\text{Span}(\phi_i; 1\leq i\leq M)= S$
and 
$$
\int \phi_i \phi_j d\mu = \delta_{ij},
$$
and define the matrix $\bW$ by 
$$ \bW_{ij} = \langle L(\cdot, \bz_i), \phi_j\rangle,$$
where $\langle \cdot,\cdot\rangle$ denotes the inner product of $L^2(\mu)$.
By definition of the Fredholm determinant for finite rank operators
\cite[Section VII.1 or Theorem VII.3.2]{GoGoKa90}, it comes
$$ \det (\cI+\cQ_\cZ) = \det \bpo \delta_{jk} + \langle \cQ_\cZ \phi_j,
\phi_i \rangle \bpc_{1\leq i,j \leq n}.$$
Since
$$
\langle \cQ_\cZ \phi_j, \phi_i \rangle = \sum_{m,n=1}^M \bW_{nj}\bW_{mi}[\bL_{\cZ}^{-1}]_{mn}
$$
it comes
$$ \det(\cI + \cQ_Z) = \det (\bI + \bW^T \bL_\cZ^{-1} \bW).$$
Applying the classical matrix determinant lemma, it comes
$$ \det(\cI + \cQ_\cZ) = \frac{\det (\bL_\cZ + \bW \bW^T)}{\det \bL_\cZ}.$$
We finally remark that
\begin{eqnarray*}
[\bW \bW^T]_{ij} &=& \sum_{k=1}^M \langle L(\bz_i,\cdot), \phi_k\rangle\,
\langle L(\bz_j,\cdot), \phi_k\rangle\\
&=& \left\langle \sum_{k=1}^M \langle L(\bz_i,\cdot), \phi_k\rangle
  \phi_k,\, L(\bz_j,\cdot)\right\rangle\\
&=& \langle L(\bz_i,\cdot),\, L(\bz_j,\cdot)\rangle.
\end{eqnarray*}
 \end{proof}
\fi

\iflong
\begin{proof}(of Proposition~\ref{p:boundsContinuous})
We first prove the right inequality in
\eqref{eq:boundDetContinuous}. From \cite[Section VII.7]{GoGoKa90},
using \eqref{eq:traceL}, it holds
\begin{equation}
\label{e:continuous:tool1}
\det(\cI+\cL) = 1 + \sum_{k=1}^\infty \frac{1}{k!} \int \det \big(\big(
L(\bx_i,\bx_j) \big)\big)d\mu(\bx_1)\dots d\mu(\bx_k).
\end{equation}
We now apply the same argument as in the proof of the finite case (proof of
Proposition~\ref{p:boundsFinite}). Denoting $\bL_{\cX\cY} = \det
((L(\bx_i,\by_i)))$ and $\bL_\cX = ((\det L(\bx_i,\bx_j))$, we know
from the positive definiteness of the kernel $L$ that $\bL_\cX -
\bL_{\cX\cZ}\bL_\cZ^{-1}\bL_{\cZ\cX}$ is positive semidefinite, which
yields
$$
\det \bL_\cX \geq \det \bL_{\cX\cZ}\bL_\cZ^{-1}\bL_{\cZ\cX}.
$$
Plugging this into \eqref{e:continuous:tool1} yields the right inequality in \eqref{eq:boundDetContinuous}. 

Upon noting that 
$$ \tr(\cQ_\cZ) = \sum_{mn}[\bL_\cZ^{-1}]_{uv} \int
L(\bx,\bz_u)L(\bz_v,\bx) d\mu(\bx) = \tr(\bL_\cZ^{-1}\bPsi),$$
the proof of the left inequality in \eqref{eq:boundDetContinuous}
follows the lines of the proof of Proposition~\ref{p:boundsFinite},
since the main tool \cite[Theorem 1]{SeSi75} is valid for any
trace-class operators.
\end{proof}
\fi

\section{Experiments}
\label{s:experiments}

\subsection{A toy Gaussian continuous experiment}
\label{s:gaussianExperiment}

In this section, we consider a DPP on $\mathbb{R}$, so that the bounds
derived in Section~\ref{s:continuous} apply. As in \cite[Section 5.1]{AFAT14}, we take the base
measure to be proportional to a Gaussian, i.e. its density is $ \mu'(x) = \kappa \cN(x \vert 0, (2\alpha)^{-2}).$
We consider a squared exponential kernel $ L(x,y) = \exp\left(-\eps^2\Vert x-y\Vert^2\right)$.
In this particular case, the spectral decomposition of operator $\cL$
is known \cite{FaMc12}\footnote{We follow the parametrization of
  \cite{FaMc12} for ease of reference.}: the eigenfunctions of $\cL$ are scaled Hermite
polynomials, while the eigenvalues are a geometrically decreasing
sequence. This 1D Gaussian-Gaussian example is interesting for two
reasons: first, the spectral decomposition of $\cL$ is known, so that we can sample exactly
from the corresponding DPP \cite{HKPV06} and thus generate synthetic
datasets. Second, the Fredholm determinant $\det(\cI+\cL)$ in this
special case is a q-Pochhammer symbol, and can thus be efficiently
computed\footnote{\url{http://docs.sympy.org/latest/modules/mpmath/functions/qfunctions.html\#q-pochhammer-symbol}},
which allows for comparison with ``ideal'' likelihood-based methods,
to check the validity of our MCMC sampler, for instance. We emphasize
that these special properties are not needed for the inference methods
in Section~\ref{s:learning}, they are simply useful to demonstrate
their correctness. 

We sample a synthetic dataset using $(\kappa, \alpha, \eps) =
(1000,0.5,1)$, resulting in $13$ points shown in red in
Figure~\ref{f:dataAndOptimizedInputs}. Applying the variational inference method of
Section~\ref{s:variational}, jointly optimizing in $\cZ$ and
$\theta=(\kappa,\alpha,\eps)$ using the CMA-ES optimizer \cite{Han06},
yields poorly consistent results: $\kappa$ varies over several orders
of magnitude from one run to the other, and relative errors for
$\alpha$ and $\eps$ go up to $100\%$ (not shown). We thus investigate
the identifiability of the parameters with the retrospective MH of
Section~\ref{s:MH}. To limit the range of $\kappa$, we choose for
$(\log\kappa,\log\alpha,\log\eps)$ a wide uniform prior
over  
$$[200, 2000]\times [-10,10]\times[-10,10].$$
We use a Gaussian proposal, the covariance matrix of
which is adapted on-the-fly \cite{HaSaTa01} so as to reach $25\%$ of
acceptance. We start each iteration with $m=20$ pseudo-inputs, and
increase it by $10$ and re-optimize when the acceptance decision
cannot be made. Most iterations could be made with $m=20$, and the maximum number of inducing inputs required in our
run was $80$. We show the results of a run of length $10\,000$ in
Figure~\ref{f:continuous}. Removing a burn-in sample of size $1000$, we show the resulting
marginal histograms in Figures~\ref{f:continuous:margKappa},
\ref{f:continuous:margAlpha}, and
\ref{f:continuous:margEps}. Retrospective MH and the ideal MH agree. The
prior pdf is in green. The posterior marginals of $\alpha$ and $\eps$
are centered around the values used for simulation, and are very
different from the prior, showing that the likelihood contains
information about $\alpha$ and $\eps$. However, as expected, almost nothing is
learnt about $\kappa$, as posterior and prior roughly coincide. This is an example of the issues that come with
parametrizing $L$ directly, as mentioned in Section~\ref{s:intro}.

To conclude, we show a set of optimized pseudo-inputs $\cZ$ in black
in Figure~\ref{f:dataAndOptimizedInputs}. We also superimpose the
marginal of any single point in the realization, which is available
through the spectral decomposition of $\cL$ here \cite{HKPV06}. In this particular
case, this marginal is a Gaussian. Interestingly, the pseudo-inputs
look like evenly spread samples from this marginal. Intuitively, they
are likely to make the denominator in the likelihood
\eqref{eq:continuousDensity} small, as they represent an ideal sample
of the Gaussian-Gaussian DPP.

\begin{figure}
\centering
\subfigure[]{
\includegraphics[width=\twofig]{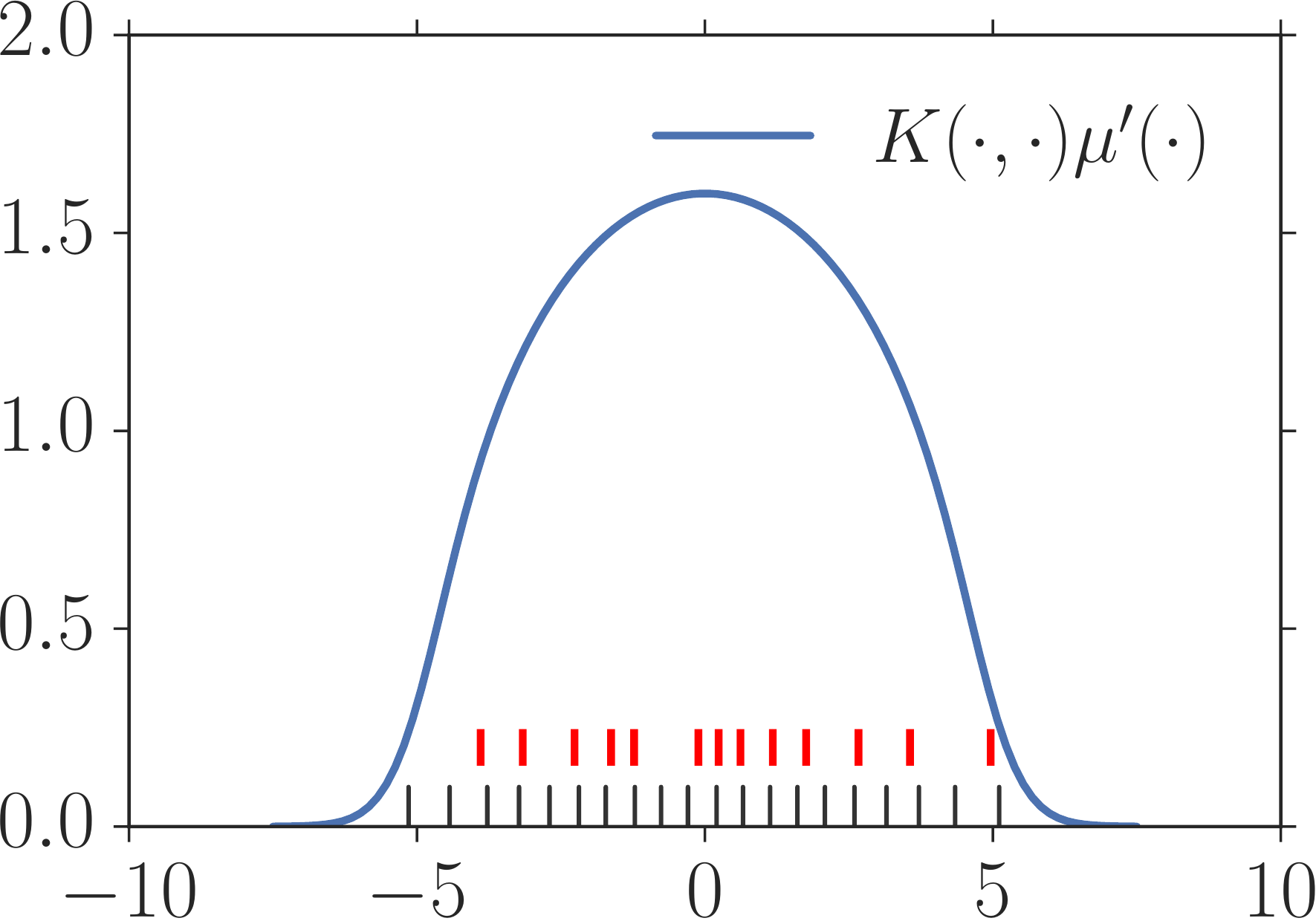}
\label{f:dataAndOptimizedInputs}
}
\subfigure{
\includegraphics[width=\twofig]{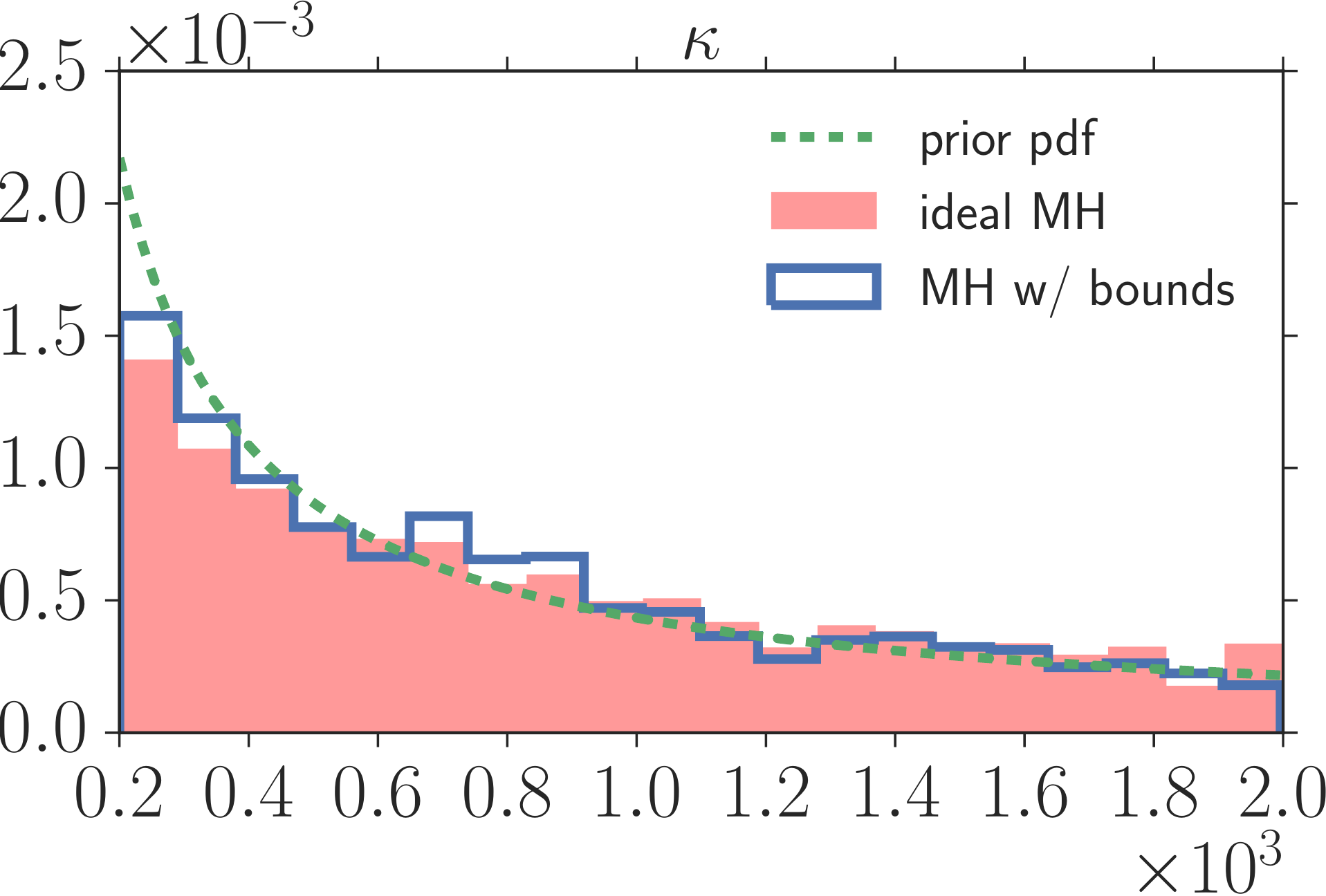}
\label{f:continuous:margKappa}
}\\
\subfigure{
\includegraphics[width=\twofig]{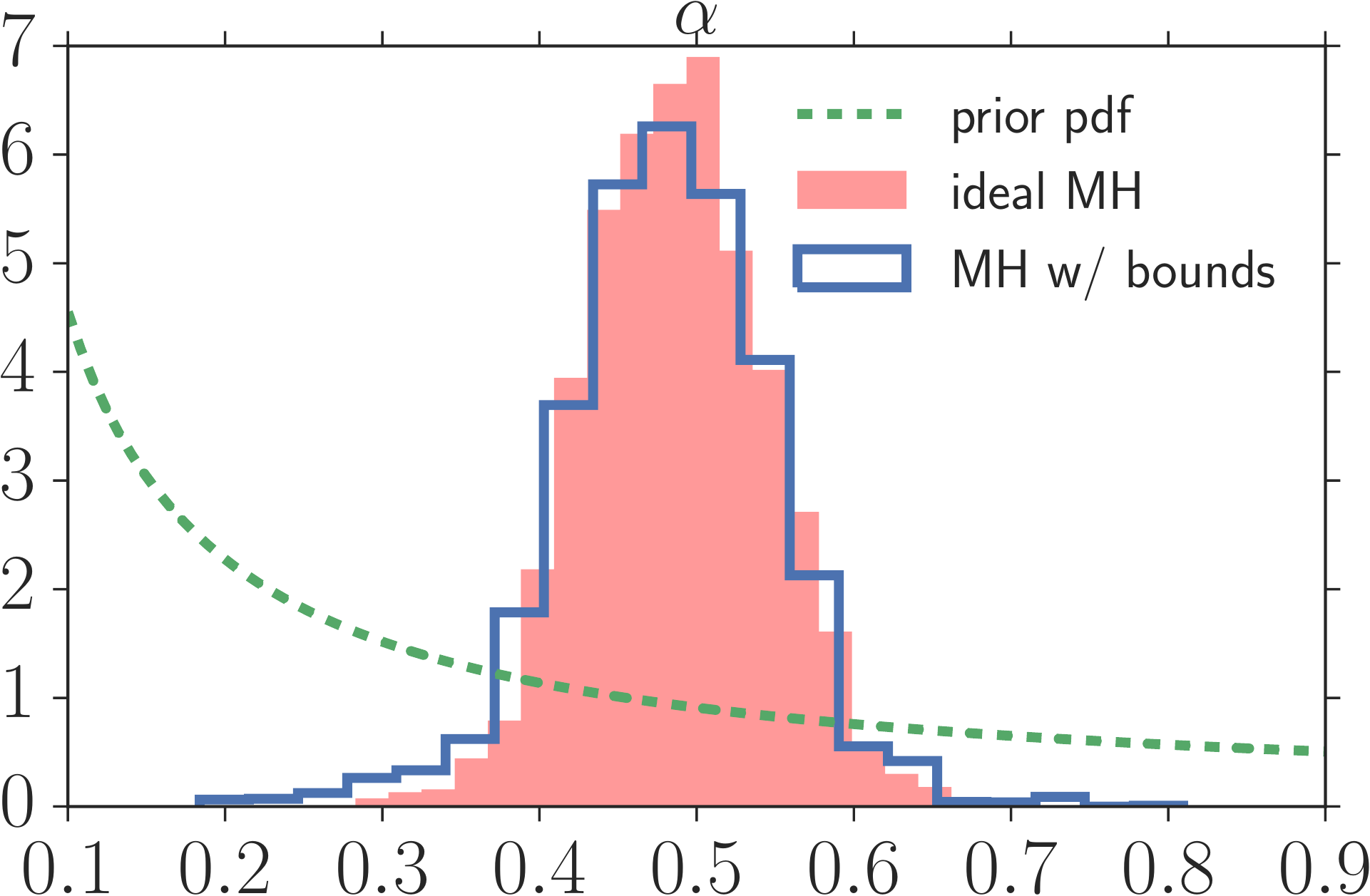}
\label{f:continuous:margAlpha}
}
\subfigure{
\includegraphics[width=\twofig]{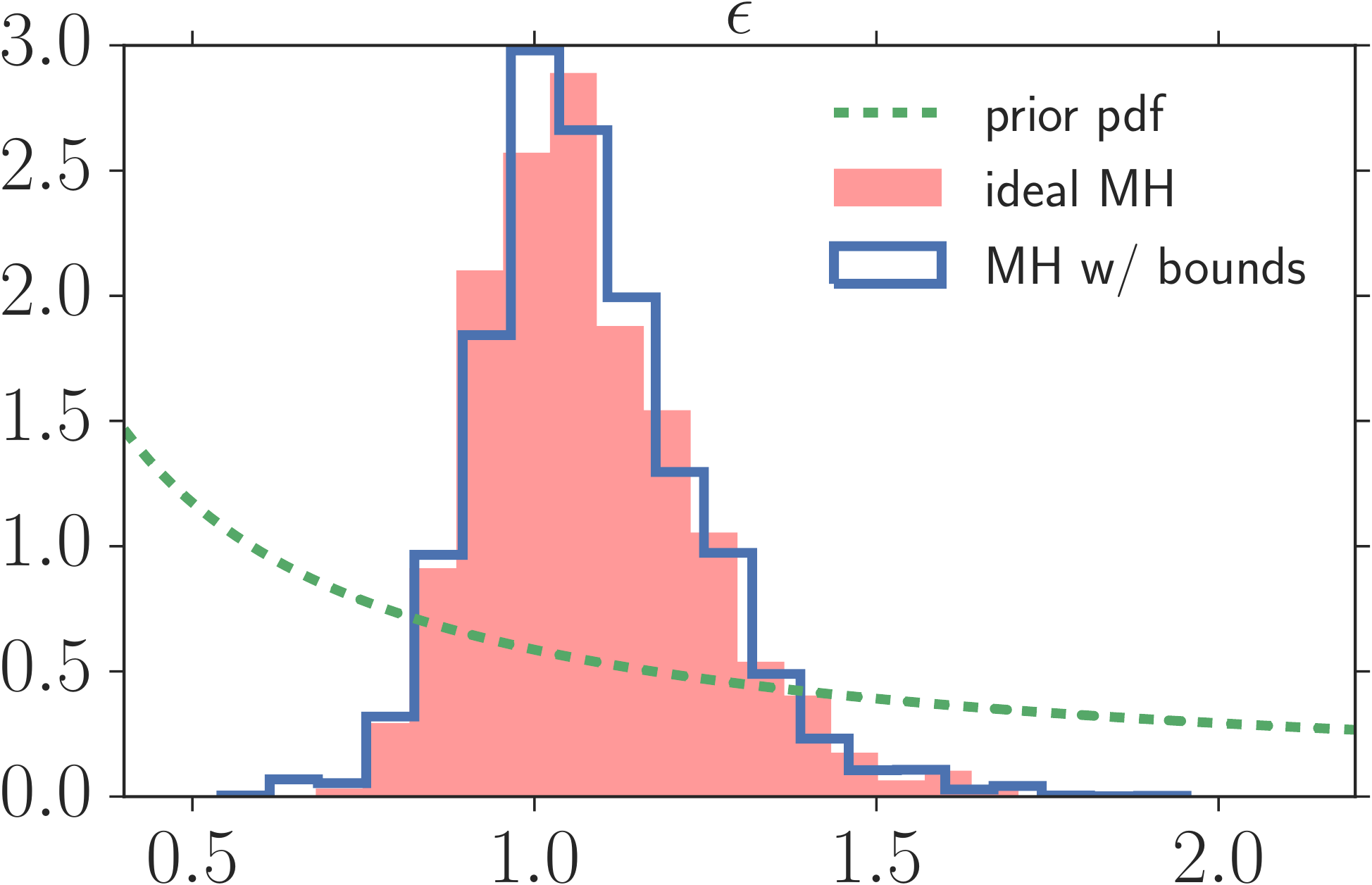}
\label{f:continuous:margEps}
}
\label{f:continuous}
\caption{Results of running adaptive Metropolis-Hastings in the 1D
  Gaussian continuous experiment of
  Section~\ref{s:gaussianExperiment}. Figure~\ref{f:dataAndOptimizedInputs}
  shows data in red, a set of optimized pseudo-inputs in black for
  $\theta$ set to the value used in the generation of the synthetic
  dataset, and the marginal of one point in the realization in
  blue. Figures~\ref{f:continuous:margKappa},
  \ref{f:continuous:margAlpha}, and
  \ref{f:continuous:margEps} show marginal histograms of
  $\kappa,\alpha,\eps$.}
\end{figure}

\subsection{Diabetic neuropathy dataset}
\label{s:nerveFiber}

Here, we consider a real dataset of spatial patterns of nerve fibers in
diabetic patients. These nerve fibers become more clustered as diabetes progresses \cite{waller01}.
The dataset consists of seven samples collected from diabetic patients at different stages of diabetic neuropathy
and one healthy subject. We follow the experimental setup used in \cite{AFAT14} 
and we split the total samples into two classes: Normal/Mildly
Diabetic and Moderately/Severely Diabetic. The first class contains three samples and the second one the remaining four. 
Figure \ref{fig:nervefiber} displays the point process data, which contain on average $90$ points per sample in the 
Normal/Mildly class and $67$ for the Moderately/Severely class. We wish to investigate the differences between these classes
by fitting a separate DPP to each class and then quantify the differences of the repulsion or overdispersion of the point process data
through the inferred kernel parameters. Paraphrasing \cite{AFAT14}, we consider a continuous DPP
on $\mathbb{R}^2$, with kernel function 
\begin{equation}
L(\bfx_i, \bfx_j) = \exp\left(- \sum_{d=1}^2 \frac{\left(x_{i,d}  - x_{j,d}\right)^2}{2 \sigma^{2}_d} \right), 
\end{equation}
and base measure proportional to a Gaussian
$\mu'(\bfx) = \kappa\prod_{d=1}^2\cN(x_d\vert\mu_d,\rho_d^2).$
As in \cite{AFAT14}, we quantify the overdispersion of realizations of
such a Gaussian-Gaussian DPP through the  quantities $\gamma_d =
\sigma_d/\rho_d$, which are invariant to the scaling of $\bfx$. Note
however that, strictly speaking, $\kappa$ also mildly influences repulsion.

We investigate the ability of the variational method in
Section~\ref{s:variational} to perform approximate maximum likelihood training over the 
kernel parameters $\theta=(\kappa, \sigma_1,\sigma_2, \rho_1,\rho_2)$. Specifically, we   
wish to fit a separate continuous DPP to each class by jointly
maximizing the variational lower bound over $\theta$ and the inducing
inputs $\cZ$ using gradient-based optimization. Given that the number
of inducing variables determines the amount of the approximation, or
{\em compression} of the DPP model, we examine different settings for
this number and see whether the corresponding trained models 
provide similar estimates for the overdispersion measures. Thus, we train the DPPs under different approximations having 
$m \in \{50,100,200,400,800,1200\}$ inducing variables and display the
estimated overdispersion measures in
Figures~\ref{f:nerveFiber:a} and \ref{f:nerveFiber:b}. These estimated
measures converge to coherent values as $m$ increases. They show a
clear separation between the two classes, as also found in
\cite{AFAT14,waller01}. Furthermore,
Figures~\ref{f:nerveFiber:c} and \ref{f:nerveFiber:d} show the values of
the upper and lower bounds on the log likelihood, which as expected,
converge to the same limit as $m$ increases. We point out that the overall optimization of the 
variational lower is relatively fast in our MATLAB implementation. For instance, it takes 
24 minutes for the most expensive run where $m=1200$ to perform
$1\,000$ iterations until convergence. Smaller values of $m$ yield
significantly smaller times. 

Finally, as in Section~\ref{s:gaussianExperiment}, we comment on the
optimized pseudo-inputs. Figure \ref{fig:nerverFiberInducing} displays
the inducing points at the end of a converged run of variational
inference for various values of $m$. Similarly to
Figure~\ref{f:dataAndOptimizedInputs}, these pseudo-inputs are placed in remarkably
neat grids and depart significantly from their initial locations.       

\begin{figure*}[!htb]
\vskip 0.2in
\begin{center}
\begin{tabular}{c}
{\includegraphics[width=\textwidth]  
{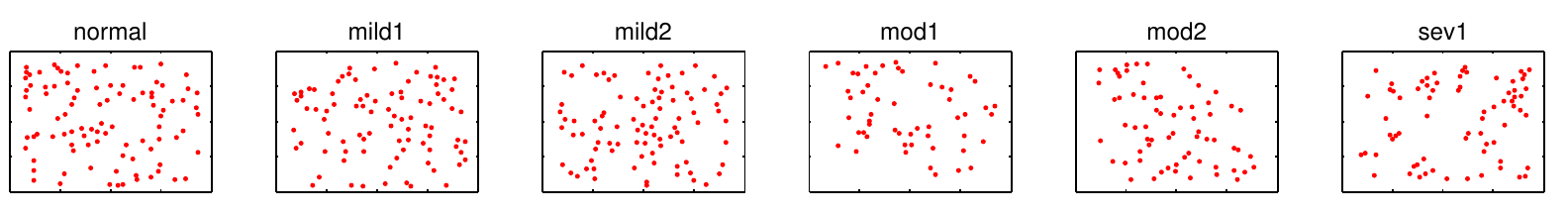}}
\end{tabular}
\caption{Six out of the seven nerve fiber samples. The first three 
samples (from left to right) correspond to a Normal Subject and two Mildly Diabetic Subjects, respectively. 
The remaining three samples correspond to a Moderately Diabetic Subject and two Severely
Diabetic Subjects.} 
\label{fig:nervefiber}
\end{center}
\vskip -0.2in
\end{figure*}

\begin{figure*}[!htb]
\vskip 0.2in
\begin{center}
\subfigure[]{
\includegraphics[width=0.22\textwidth]{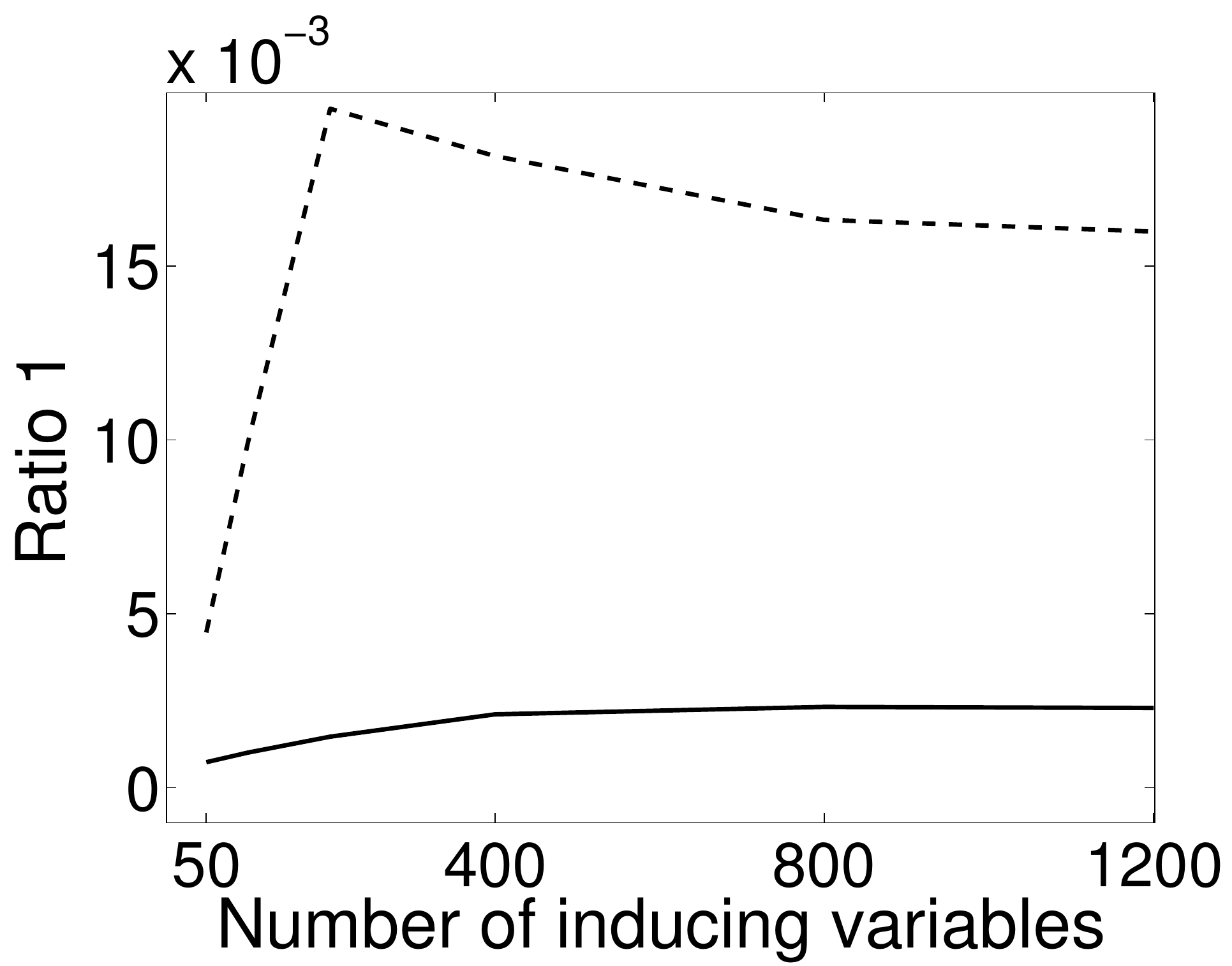}
\label{f:nerveFiber:a}
}
\subfigure[]{
\includegraphics[width=0.22\textwidth]{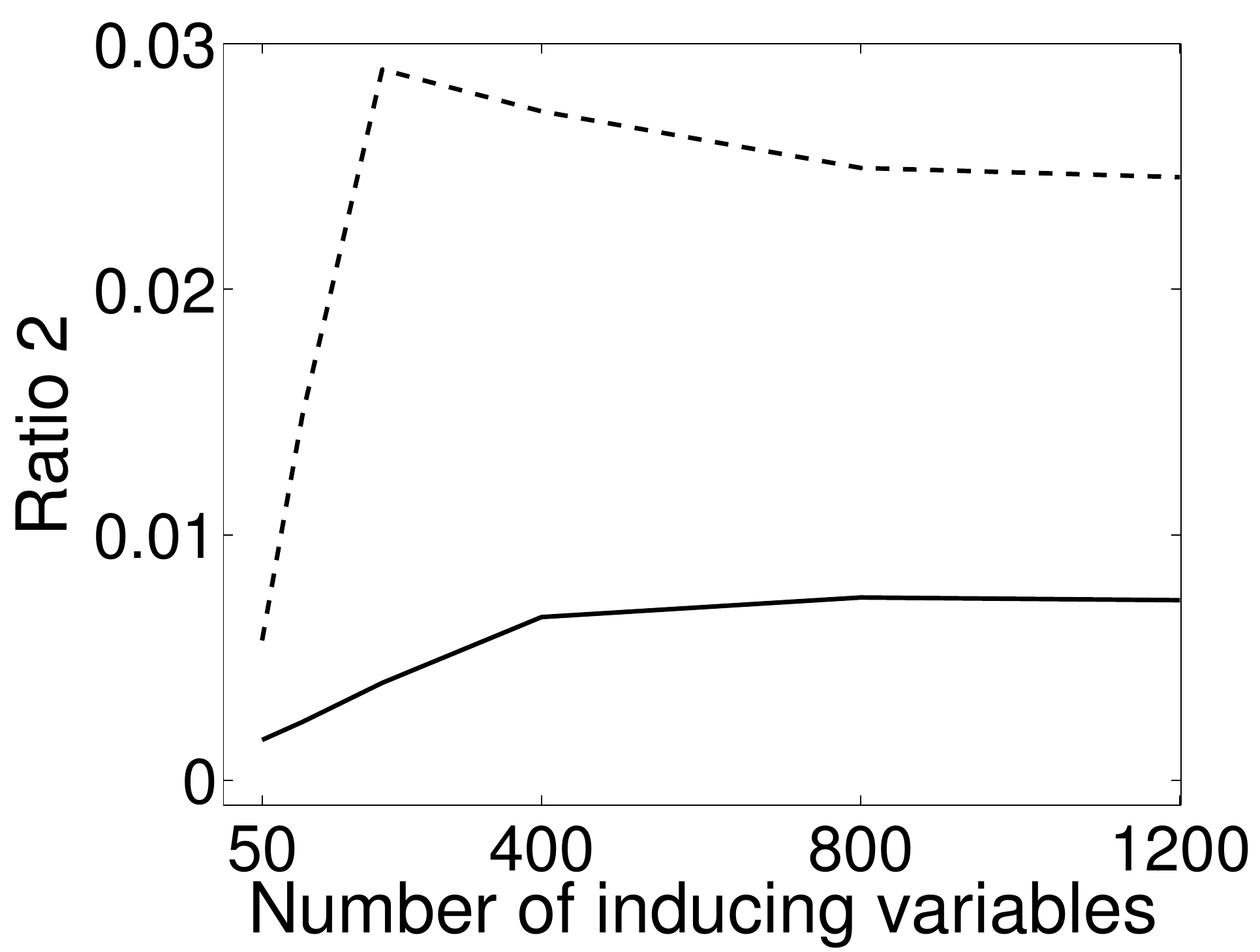}
\label{f:nerveFiber:b}
}
\subfigure[]{
\includegraphics[width=0.22\textwidth]{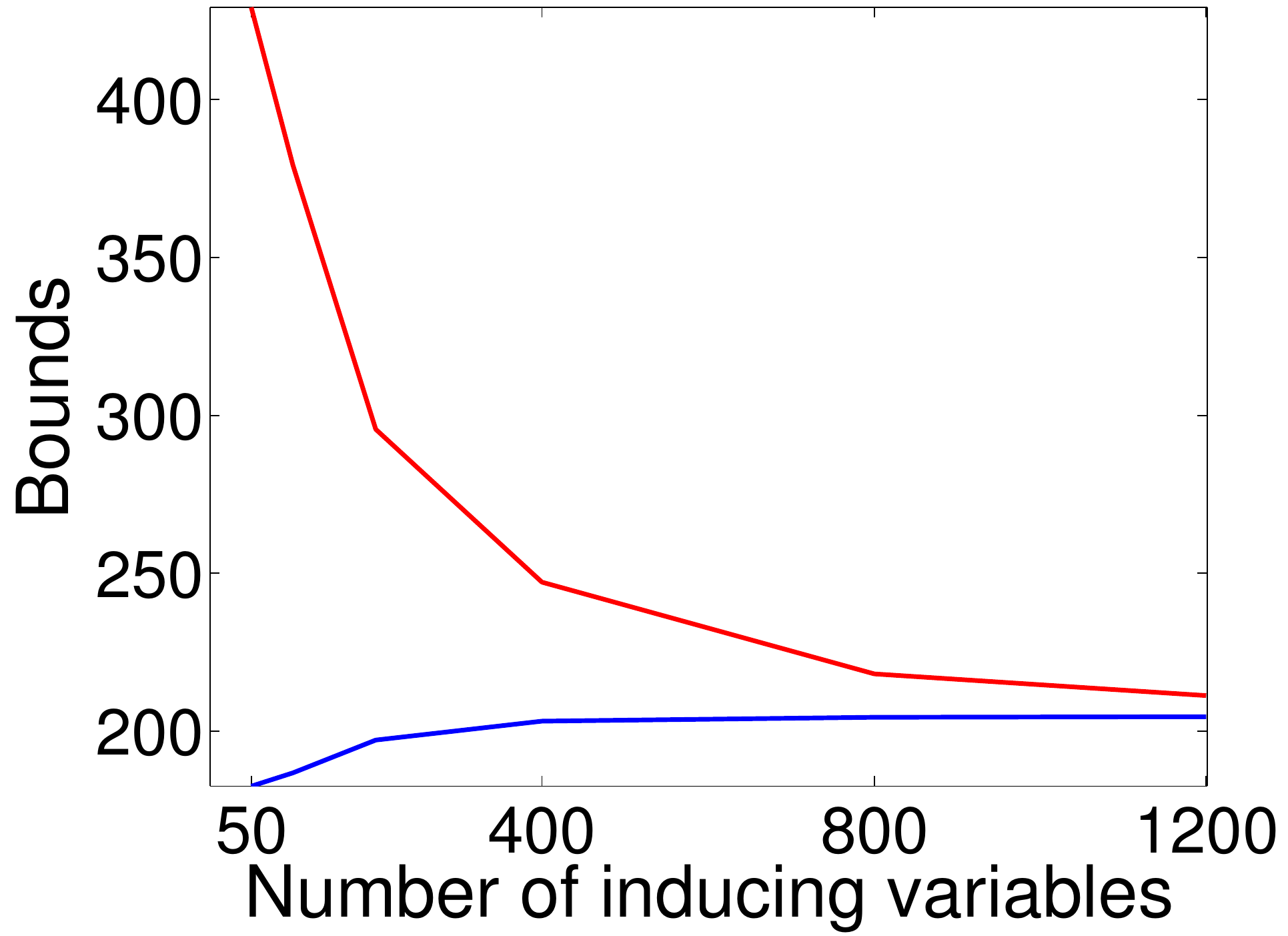}
\label{f:nerveFiber:c}
}
\subfigure[]{
\includegraphics[width=0.22\textwidth]{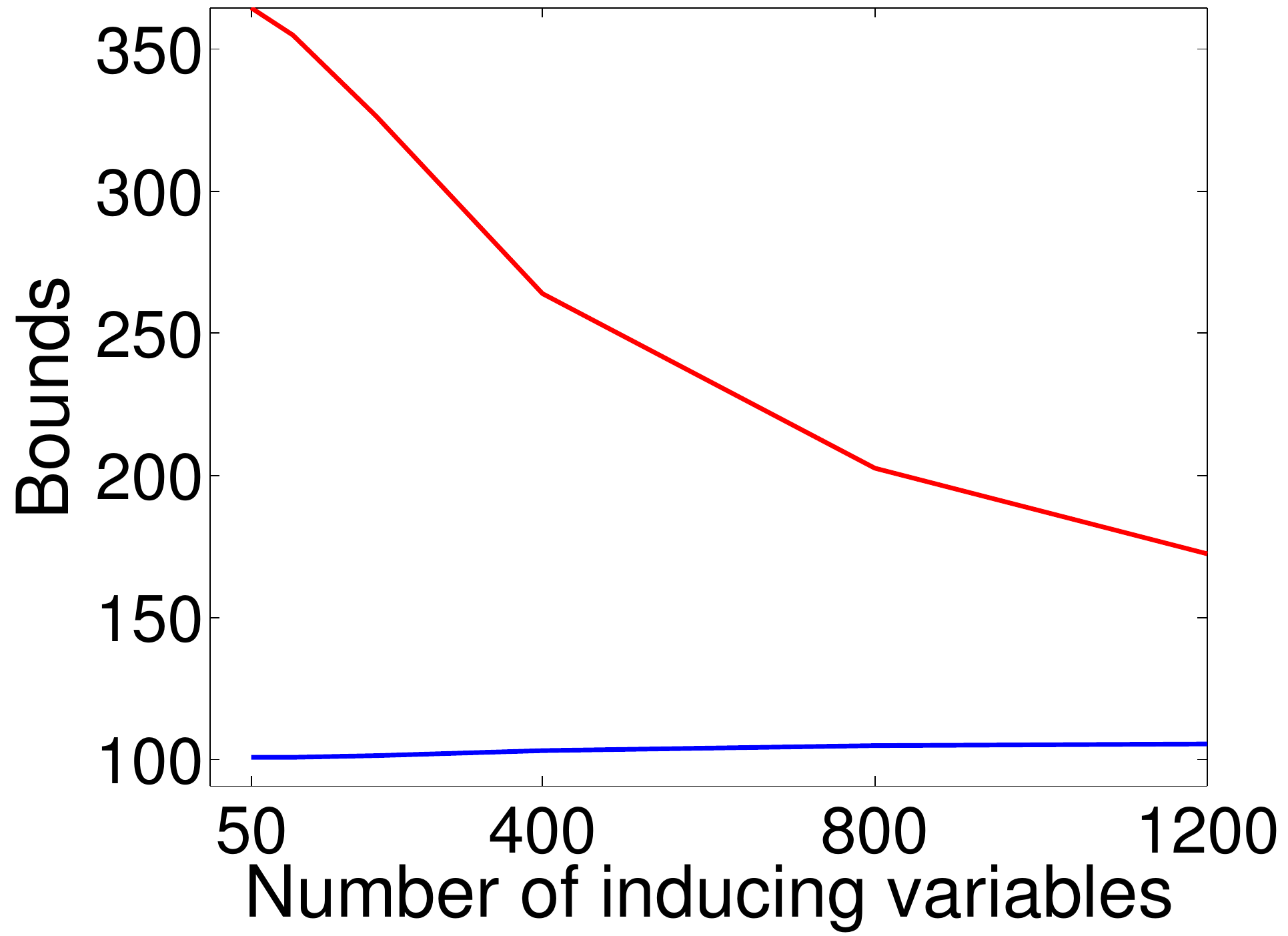}
\label{f:nerveFiber:d}
}
\caption{Figures~\ref{f:nerveFiber:a} and \ref{f:nerveFiber:b} show the evolution of the estimated overdispersion measures $\gamma_1$ and $\gamma_2$ as functions 
of the number of inducing variables used. The dotted black lines correspond to the Normal/Mildly
Diabetic class while the solid lines to the Moderately/Severely Diabetic class. Figure~\ref{f:nerveFiber:c} shows the upper bound (red) 
and the lower bound (blue) on the log likelihood as functions of the number of inducing variables for the Normal/Mildly
Diabetic class while the Moderately/Severely Diabetic case is shown in Figure~\ref{f:nerveFiber:d}.} 
\label{fig:nervefiberParametersBounds}
\end{center}
\vskip -0.2in
\end{figure*}

\begin{figure*}[!htb]
\vskip 0.2in
\begin{center}
 {\includegraphics[width=\textwidth]  
{\figdir/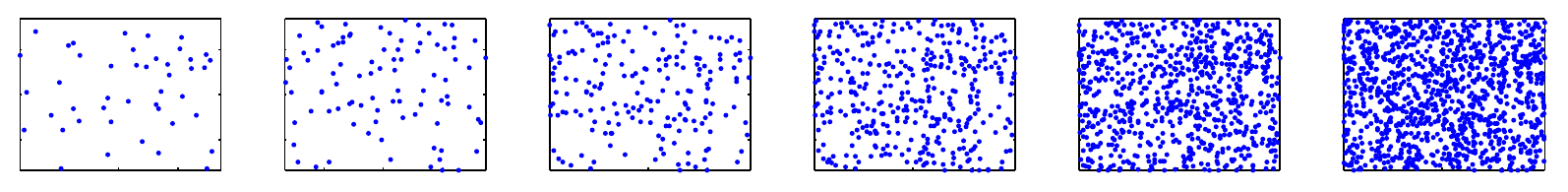}} 
{\includegraphics[width=\textwidth]  
{\figdir/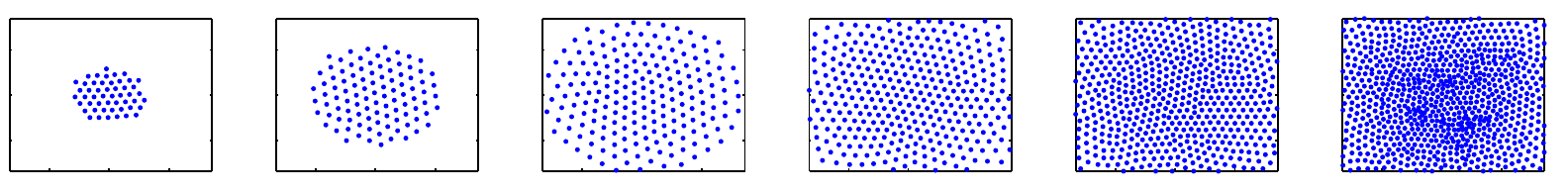}}
\caption{We illustrate the optimization over the inducing inputs
  $\cZ$ for several values of  $m \in \{50,100,200,400,800,1200 \}$ in
  the DPP of Section~\ref{s:nerveFiber}. We consider the Normal/Mildly
  diabetic class. The panels in the top row show the initial inducing
  input locations for various values of $m$, while the corresponding panels in the bottom row 
show the optimized locations.} 
\label{fig:nerverFiberInducing}
\end{center}
\vskip -0.2in
\end{figure*}

\section{Discussion}
\label{s:discussion}

We have proposed novel, cheap-to-evaluate, nonspectral bounds on the
determinants arising in the likelihoods of DPPs, both finite and continuous. We have shown how to use these bounds to infer the
parameters of a DPP, and demonstrated their use for
expensive-but-exact MCMC and cheap-but-approximate variational
inference. In particular, these bounds have some degree of freedom
--~the pseudo-inputs~--, which we optimize so as to tighten the bounds. This optimization step is
crucial for likelihood-based inference of parametric DPP models, where bounds have to adapt to the point where the likelihood
is evaluated to yield decisions which are consistent with the ideal
underlying algorithms. In future work, we plan to investigate connections of our bounds with
the quadrature-based bounds for Fredholm determinants of \cite{Bor10}. We also plan to consider
variants of DPPs that condition on the number of points in the
realization, to put joint priors over the within-class distributions
of the features in classification problems, in a manner related to
\cite{ZoAd12}. In the long term, we will investigate connections between
kernels $L$ and $K$ that could be made without spectral knowledge, to address
the issue of replacing $L$ by $K$.

\subsection*{Acknowledgments}
The authors would like to thank Adrien Hardy for useful
  discussions and Emily Fox for kindly providing access to the
  diabetic neuropathy dataset.

\iflong
\appendix 

\section{On the connection to variational sparse GPs}

Here, we provide further details about how the variational lower bound
for DPPs over finite sets in Proposition~\ref{p:boundsFinite} can be obtained by the variational approach to sparse GPs \cite{titsias}. As mentioned 
in Section~\ref{s:svi}, it holds that  
$$
 \frac{1}
 {\text{det}(\bL +\bI)}
 = \left( \int  \mathcal{N}(\bff| \bL_{\cY\cZ}[\bL_\cZ]^{-1}\bfu, \bL - \bfQ) \mathcal{N}(\bfu|{\bf 0}, \bL_\cZ)  e^{ - \frac{1}{2} \bff^T \bff }  d \bff d \bfu \right)^2. 
$$
Taking logarithms yields 
$$
\log  \frac{1}
 {\text{det}(\bL +\bI)}
 = 2  \log  \int  \mathcal{N}(\bff| \bL_{\cY\cZ}[\bL_\cZ]^{-1}\bfu, \bL - \bfQ) \mathcal{N}(\bfu|{\bf 0}, \bL_\cZ)  e^{ - \frac{1}{2} \bff^T \bff }  d \bff d \bfu.
$$
A lower bound to the likelihood \eqref{e:discreteDensity} can thus be
obtained if we bound
$$ 
\mathcal{F} = \log \int \mathcal{N}(\bff| \bL_{\cY\cZ}[\bL_\cZ]^{-1}\bfu, \bL - \bfQ) \mathcal{N}(\bfu|{\bf 0}, \bL_\cZ)  e^{ - \frac{1}{2} \bff^T \bff }  d \bff d \bfu.
$$
This has a similar functional form with the marginal likelihood in a standard GP regression model: 
$e^{ - \frac{1}{2} \bff^T \bff }$ plays the role of an unnormalized Gaussian likelihood where the observation 
vector is equal to zero and the noise variance is equal to one. 
To lower bound the above we can consider the variational distribution  $q(\bff,\bfu) = \mathcal{N}(\bff| \bL_{\cY\cZ}[\bL_\cZ]^{-1}\bfu, \bL - \bfQ) q(\bfu)$  
and apply Jensen's inequality so that 
$$ 
\mathcal{F} \geq \int \mathcal{N}(\bff| \bL_{\cY\cZ}[\bL_\cZ]^{-1}\bfu, \bL - \bfQ) q(\bfu) \log 
\frac{ \mathcal{N}(\bfu|{\bf 0}, \bL_\cZ)  e^{ - \frac{1}{2} \bff^T \bff } }
{q(\bfu)} d \bff d \bfu,
$$
where the term $\mathcal{N}(\bff| \bL_{\cY\cZ}[\bL_\cZ]^{-1}\bfu, \bL
- \bfQ)$ cancels out inside the logarithm. This can be written
as
$$ 
\mathcal{F} \geq \int q(\bfu)
\left\{ - \frac{1}{2} \int \mathcal{N}(\bff| \bL_{\cY\cZ}[\bL_\cZ]^{-1}\bfu, \bL - \bfQ) \bff^T \bff d \bff  
 +  \log \frac{ \mathcal{N}(\bfu|{\bf 0}, \bL_\cZ) }
{q(\bfu)}  \right \} d \bfu.
$$
Further, given that   
$$
\int \mathcal{N}(\bff| \bL_{\cY\cZ}[\bL_\cZ]^{-1}\bfu, \bL - \bfQ) \bff^T \bff d \bff = \boldsymbol{\alpha}^T \boldsymbol{\alpha} 
+  \text{tr}(\bL - \bfQ),$$
 where $\boldsymbol{\alpha} = \bL_{\cY\cZ}[\bL_\cZ]^{-1}\bfu$, the bound can be written as 
$$ 
\mathcal{F} \geq \int q(\bfu) \log 
\frac{ \mathcal{N}(\bfu|{\bf 0}, \bL_\cZ)  e^{ - \frac{1}{2} \boldsymbol{\alpha}^T \boldsymbol{\alpha} } }
{q(\bfu)} d \bfu   - \frac{1}{2} \text{tr}(\bL - \bfQ).
$$
Now if we analytically maximize w.r.t. $q(\bfu)$, under the constraint
that $q(\bfu)$ is a distribution, we obtain 
$$q(\bfu) = \frac{\mathcal{N}(\bfu|{\bf 0}, \bL_\cZ)  e^{ - \frac{1}{2} \boldsymbol{\alpha}^T \boldsymbol{\alpha} }}{
\int \mathcal{N}(\bfu|{\bf 0}, \bL_\cZ)  e^{ - \frac{1}{2} \boldsymbol{\alpha}^T \boldsymbol{\alpha} }  d \bfu}.$$
Plugging this optimal $q$ back into the bound, we obtain 
$$
\mathcal{F} \geq \log \int \mathcal{N}(\bfu|{\bf 0}, \bL_\cZ)  e^{ - \frac{1}{2} \boldsymbol{\alpha}^T \boldsymbol{\alpha} }  d \bfu
- \frac{1}{2} \text{tr}(\bL - \bfQ).
$$ 
After computing the Gaussian integral w.r.t.\ $\bfu$, the r.h.s. reduces to the logarithm of the DPP bound for 
the finite case, see Proposition~\ref{p:boundsFinite}.

\section{ $\Psi$ matrix for Gaussian kernels} 

We give here more details on the Gaussian kernel with Gaussian base
measure used in the experimental Section~\ref{s:experiments}. We use
the notation of Section~\ref{s:nerveFiber}. The kernel is
$$
L(\bfx_i, \bfx_j) = e^{ - \sum_{d=1}^D \frac{ \left(x_{i,d}  - x_{j,d} \right)^2}{2 \sigma_d^2}  }, 
$$
with Gaussian base measure having density 
$$
\mu'(\bfx) = \kappa  \prod_{d=1}^D \frac{1}{\sqrt{2 \pi \rho_d^2}} e^{ - \frac{1}{2 \rho^2_d} \left(x_d - \mu_d\right)^2}.
$$ 
In this Gaussian-Gaussian case, the $\Psi$ matrix defined in Proposition~\ref{p:boundsContinuous} can be analytically computed: the $ij$-th element is given by 
\begin{align}
[\Psi]_{i j} & = \int_{\mathbbm{R}^D} L(\bfz_i, \bfx) L(\bfx, \bfz_j) d \mu(\bx) = \kappa \prod_{d=1}^D 
\frac{e^{- \frac{1}{4} \sigma_d^{-2} ( z_{i,d} - z_{j,d})^2 - \frac{\sigma^{-2}_d (\mu_d - \bar{z}_d)^2}{2 \sigma_d^{-2} \rho_d^2 + 1 } }}
{\left(2 \sigma_d^{-2} \rho^2_d + 1 \right)^{\frac{1}{2}}}, 
\end{align}
where $\bar{z}_d  = \frac{ z_{i,d} + z_{j,d}}{2}$.
\fi
\bibliographystyle{unsrt}
\bibliography{stats,learning,stats2}

\end{document}